\documentclass[11pt]{article}
\usepackage{fullpage}
\usepackage{ifthen}
\usepackage{xspace}
\usepackage{amsmath, amssymb, amsfonts}
\usepackage{amsthm}
\usepackage{bm}
\usepackage[english]{babel}
\usepackage[showdeletions]{color-edits}
\usepackage{algorithm}
\usepackage{algorithmic}

\begin{document}

\newcommand{\tr}{\textsf{triangle}}
\newcommand{\li}{\textsf{line}}
\newcommand{\xx}[3]{{#1}^{#2}_{#3}}
\newcommand{\mc}[1]{\mathcal{#1}}
\newcommand{\BRi}[1]{\mathcal{B}^{R_i}(#1)}
\newcommand{\bK}{\mc{\bar K}}
\newcommand{\bq}{\bar q}
\newcommand{\Ex}[1]{\mathop{\mathrm{E}}\left[{#1}\right]}
\newcommand{\ex}{\mathop{\mathrm{E}}}
\newcommand{\pr}[1]{\Pr\left[{#1}\right]}

\newtheorem{theorem}{Theorem}
\newtheorem{claim}{Claim}
\newtheorem{corollary}[theorem]{Corollary}
\newtheorem{proposition}[theorem]{Proposition}
\newtheorem{lemma}[theorem]{Lemma}
\newtheorem{property}{Property}
\newtheorem{fact}[theorem]{Fact}

\newtheorem{definition}{Definition}
\newtheorem{example}{Example}
\newtheorem{assumption}[theorem]{Assumption}

\newtheorem{remark}{Remark}

\addauthor{pc}{red}

\title{Bounds on the Price of Anarchy for a More General Class of Directed Graphs in Opinion Formation Games}

\author{%
  Po-An Chen\thanks{%
    Institute of Information Management, National Chiao Tung University, Taiwan.
    Email: poanchen@nctu.edu.tw.
  }
  \footnote{Corresponding author.}
  \and
  Yi-Le Chen\thanks{%
    Institute of Information Management, National Chiao Tung University, Taiwan.
    Email: jairachen78@gmail.com.
  }
  \and
  Chi-Jen Lu\thanks{%
    Institute of Information Science, Academia Sinica, Taiwan.
    Email: cjlu@iis.sinica.edu.tw. Also with Department of Computer Science, National Chiao-Tung University, Taiwan.
  }
}
\date{}

\maketitle

\begin{abstract}
In opinion formation games with directed graphs, a bounded
price of anarchy is only known for weighted Eulerian graphs. Thus, we
bound the price of anarchy for a more general class of directed graphs
with conditions intuitively meaning that each node does not influence the
others more than she is influenced, where the bounds depend on such
difference (in a ratio). We also show that there exists an example just
slightly violating the conditions with an unbounded price of anarchy.
%
\end{abstract}



\section{Introduction}
In a society or community, individuals and their relationships form a social network.
For some matter that each individual gets to express her own opinion about,
individuals influence each others regarding it through such a social network.
For example, the matter can be adopting a new innovation/product, and the individuals are potential users/consumers while their opinions could be tendency to adopt the innovation or the preference toward the product;
for some political issue that may need public consensus, the public have different opinions or thoughts about it.
In any case, an individual is often affected by her friends/neighbors in the social network when making up her mind.
The opinion forming process in a social network can be naturally thought as an opinion influencing and updating dynamics.
This has attracted researchers' interest a while ago in mathematical sociology, and more recently in theoretical computer science.

DeGroot \cite{degroot} modeled the opinion formation process by associating each individual with a numeric-value opinion
and letting the opinion be updated by a weighted average of the opinions from her friends and her own,
where the weights represent how much she is influenced by her friends.
This update dynamics will converge to a fixed point in which all individuals hold the same opinion, i.e., a consensus.
However, we can easily observe that in the real world, the consensus is difficult to reach.
Friedkin and Johnson \cite{friedkin} differentiated an \emph{expressed opinion} that each individual in the networks updates over time
from an \emph{internal opinion} that each individual is born with and stays unchanged.
Thus, an individual would be always influenced by her inherent belief,
and the dynamics converges to an unique equilibrium, which may not be a consensus.

Bindel et al. \cite{bindel} viewed the updating rule mentioned above equivalently as each player updating her expressed opinion to minimize her quadratic individual cost function,
which consists of the disagreement between her expressed opinion and those of her friends, and the difference between her expressed and internal opinions.
They analyzed how socially good or bad the system can be at equilibrium compared to the optimum solution in terms of the \emph{price of anarchy} \cite{koutsoupias:papadimitriou:anarchy},
i.e., the ratio of the social cost at the worst equilibrium to the optimal social cost. (Notions of equilibria will be given in Section~2.)
The price of anarchy is at most $9/8$ in undirected graphs and is unbounded in directed graphs (due to a star graph with a center only influencing the others but is not influenced at all by the others, or even directed bounded-degree trees with degrees high enough).
Nevertheless, a bounded price of anarchy can be obtained for weighted Eulerian graphs, in particular, a tight upper bound of $2$ for directed cycles and an upper bound of $d+1$ for $d$-regular graphs.
Another work closely related to that of Bindel et al. is by Bhawalkar et al. \cite{bhawalkar}.
The individual cost functions are assumed to be ``locally-smooth" in the sense of \cite{roughgarden3} and may be more general than quadratic, for example, convex. The price of anarchy for undirected graphs with quadratic cost functions is at most $2$.
They also allowed social networks to change by letting players choose $k$-nearest neighbors through opinion updates and bounded the price of anarchy.
On the other hand, Chierichetti et al. \cite{chierichetti} considered the games with discrete preferences,
where an expressed and internal opinions are chosen from a discrete set and distances measuring ``similarity" between opinions correspond to costs.

When graphs are directed, a bounded price of anarchy is only known for weighted Eulerian graphs where the total incoming weights equal to the total outgoing weights at each node \cite{bindel,bhawalkar}, which may seem rather restricted.
Thus, we are interested \emph{to bound the price of anarchy for games with directed graphs more general than weighted Eulerian graphs} (even with just quadratic individual cost functions) in this article. Note that although the result of \cite{bhawalkar} is indeed for directed graphs and gives bounded price of anarchy, their setting is different from ours. In their model, the weights on the $k$ neighbors are uniform, and more importantly, the $k$ neighbors of a node is not fixed, as its action in the game includes choosing its $k$ neighbors in addition to its opinion. Therefore, it is related to but different from what we propose to tackle here.
We first bound the price of anarchy for a more general class of directed graphs with conditions intuitively meaning that each node does not influence the others more than she is influenced by herself and the others, where the bounds depend on such influence differences (in a ratio).
This generalizes the previous results on directed graphs, and recovers and matches the previous bounds in some specific classes of (directed) Eulerian graphs.
We then show the existence of an example that just slightly violates the conditions but with an unbounded price of anarchy.
We further propose more research directions in the discussions and future work.


\section{Preliminaries}
We describe a social network as a weighted graph $(G,\mathbf{w})$ for directed graph $G=(V,E)$ and matrix $\mathbf{w}=[w_{ij}]_{ij}$.
The node set $V$ of size $n$ is the selfish players, and the edge set $E$ is the relationships between any pair of nodes.
The edge weight $w_{ij}\geq 0$ is a real number and represents how much player~$i$ is influenced by player~$j$;
note that weight $w_{ii}$ can be seen as a self-loop weight, i.e., how much player~$i$ influences (or is influenced by) herself.
Each (node) player has an internal opinion $s_i$, which is unchanged and not affected by opinion updates.
An \emph{opinion formation game} can be expressed as an instance $(G,\mathbf{w},\mathbf{s})$ that combines weighted graph $(G,\mathbf{w})$ and vector $\mathbf{s}=(s_i)_i$.
Each player's strategy is an expressed opinion $z_i$, which may be different from her $s_i$ and gets updated.
Both $s_i$ and $z_i$ are real numbers.
The individual cost function of player~$i$ is
\[C_i(\mathbf{z})=w_{ii}(z_i-s_i)^2+\sum_{j\in N(i)}w_{ij}(z_i-z_j)^2=w_{ii}(z_i-s_i)^2+\sum_j w_{ij}(z_i-z_j)^2,\]
where $\mathbf{z}$ is the strategy profile/vector and $N(i)$ is the set of the neighbors of $i$, i.e., $\{j:j\neq i,w_{ij}>0\}$.
It measures the disagreement between her expressed opinion and those of her friends, and the difference between her expressed and internal opinions.
Other functions could also be used other than the square one here, for example, convex functions \cite{bhawalkar}.
Each node minimizes her cost $C_i$ by choosing her expressed opinion $z_i$.
We analyze the game when it stabilizes, i.e., at equilibrium.

\subsubsection*{Equilibria and the Price of Anarchy}
In a (pure) Nash equilibrium $\mathbf{z}$, each player~$i$'s strategy is $z_i$ such that given $\mathbf{z}_{-i}$ (i.e., the opinion vector of all players except $i$) for any other $z'_i$,
\[C_i(z_i,\mathbf{z}_{-i}) \leq C_i(z'_i,\mathbf{z}_{-i}).\]
According to \cite{bindel,bhawalkar}, in such an equilibrium, the condition
\[z_i=\frac{w_{ii}s_i+\sum_{j\neq i}w_{ij}z_j}{w_{ii}+\sum_{j\neq i}w_{ij}} \]
holds for any player $i$. This can be shown by taking the derivative of $C_i$ w.r.t. $z_i$, setting it to $0$ for each $i$,
and solving the equality system since very player~$i$ minimizes $C_i$.
Note that $C_i$ is continuously differentiable.

The social cost function here is $C(\mathbf{z})=\sum_{i}C_i(\mathbf{z})$, the sum of the individual costs.
Nash equilibria can be far from the (centralized) social optimum in terms of a social cost \cite{NRTV07}.
To measure the (in)efficiency of equilibria, the price of anarchy \cite{koutsoupias:papadimitriou:anarchy} is defined as the ratio of the worst equilibrium's social cost to the optimal social cost.

\subsubsection*{Local Smoothness}
To bound the price of anarchy, the \emph{local smoothness} framework developed by Roughgarden and Schoppmann \cite{roughgarden3} is a promising analysis technique in algorithmic game theory.
It has been applied in \cite{bhawalkar} to obtain the price of anarchy bounds there and similar techniques have been used in many other games \cite{roughgarden,chen:dekeijzer}. The local smoothness technique is slightly different from the smoothness techniques of \cite{roughgarden,chen:dekeijzer}. The local smoothness technique is sometimes more suitable since it gives tight(er) bounds while the smoothness technique does not in some games.
We briefly summarize this technique in the following. The inequality intuitively means that summing up individual costs after some unilateral local deviations to $o_i$'s from any strategy profile $\mathbf{z}$ can still be upper bounded by a combination of the social costs $C(\mathbf{z})$ and $C(\mathbf{o})$ (where the derivative term accounts for localness), implying that $C(\mathbf{z})$ is not too far from $C(\mathbf{o})$. As in \cite{roughgarden3}, we assume that the cost function $C_i$ of each player $i$ is continuously differentiable w.r.t. her strategy.
\begin{definition}
A cost-minimization game is $(\lambda,\mu)$-locally smooth if for any strategy profiles $\mathbf{o}$ and $\mathbf{z}$,
\begin{eqnarray} \label{def:smooth}
\sum_{i} \left(C_i(z_i,\mathbf{z}_{-i})+(o_i-z_i)\frac{\partial}{\partial z_i}C_i(z_i,\mathbf{z}_{-i})\right) \leq \lambda C(\mathbf{o})+\mu C(\mathbf{z}).
\end{eqnarray}
\end{definition}
We have the following from an extension theorem of \cite{roughgarden3}.
\begin{theorem} \label{thm:extension}
If $\sigma$ is a correlated equilibrium and $\mathbf{o}$ is the social optimum in a $(\lambda,\mu)$-locally smooth game, for $\lambda>0$ and $\mu <1$, then
the correlated (thus, pure and mixed) price of anarchy, $E_{\mathbf{z}\sim\sigma}[C(\mathbf{z})]/C(\mathbf{o})$, is at most $\lambda/({1-\mu})$.
\end{theorem}
With this, our goal becomes finding suitable $\lambda$ and $\mu$ values to satisfy Inequality~(\ref{def:smooth}),
which immediately gives price-of-anarchy bounds.

\section{Bounds on the Price of Anarchy}
We first generalize the result \cite{bindel} about bounded price of anarchy for weighted Eulerian graphs to a more general class of directed graphs with conditions intuitively meaning that each node does not influence the others more than she is influenced (by herself and the others).
We then show that there exists an example that just slightly violates the conditions and gives an unbounded price of anarchy.

Recall from \cite{bindel} that in a weighted Eulerian graph, $\sum_{j\neq i}w_{ij}=\sum_{j\neq i}w_{ji}$ for every node $i$, which means that the influence each node exerts on the others is exactly the same as that it receives from the others. Here we relax the condition and obtain the following.
\begin{theorem} \label{thm:bound}
Consider any value $\epsilon>0$ and any weighted graph $(G,\mathbf{w})$ such that for any node $i$,
\begin{eqnarray} \label{thm:condition}
\frac{\sum_{j\neq i}w_{ji}}{w_{ii}+\sum_{j\neq i}w_{ij}}\leq\frac{1}{1+\epsilon}.
\end{eqnarray}
Then the price of anarchy for any opinion formation game on $(G,\mathbf{w})$ is at most $1+\frac{1}{\epsilon}$.
\end{theorem}
To make the price of anarchy bound as small as possible, one chooses $\epsilon$ to be the largest possible value such that the condition in \eqref{thm:condition} still holds for every node $i$. Roughly speaking, $\epsilon$ captures a bound on how much more (as a ratio) a node is influenced by itself and the others than it influences the others.
\begin{proof}
We want to show $(\lambda,\mu)$-local smoothness for any opinion formation game on such a weighted graph.
Let $\mathbf{z}$ be a Nash equilibrium and $\mathbf{o}$ the social optimum.
We start by plugging the definition of $C_i$ into the left-hand side of~(\ref{def:smooth}) and rearranging terms as follows
\begin{eqnarray*}
\lefteqn{\sum_{i}\left(w_{ii}(z_i-s_i)^2 + \sum_{j\neq i}w_{ij}(z_i-z_j)^2\right)+ 2\sum_{i}(o_i-z_i)\left(w_{ii}(z_i-s_i) + \sum_{j\neq i}w_{ij}(z_i-z_j)\right)}\\
&=&\sum_{i}\sum_{j\neq i}w_{ij}\left((z_i-z_j)^2+2(o_i-z_i)(z_i-z_j)\right)+ \sum_{i}w_{ii}\left((z_i-s_i)^2+2(o_i-z_i)(z_i-s_i)\right)\\
&=&\sum_{i}\sum_{j\neq i}w_{ij}\left((o_i-o_j+o_j-z_j)^2-(o_i-z_i)^2\right)+\sum_{i}w_{ii}\left((o_i-s_i)^2-(o_i-z_i)^2\right).
\end{eqnarray*}

By the Cauchy-Schwarz inequality $(x+y)^2\leq(a^2+b^2)(\frac{x^2}{a^2}+\frac{y^2}{b^2})=(1+\frac{b^2}{a^2})x^2+(1+\frac{a^2}{b^2})y^2$
for $x=o_i-o_j$, $y=o_j-z_j$, and $\frac{a^2}{b^2}=\epsilon$,
the expression above is at most
\begin{eqnarray*}
&&\sum_{i}\sum_{j\neq i}w_{ij}((1+1/\epsilon)(o_i-o_j)^2+(1+\epsilon)(o_j-z_j)^2)-\sum_{i}\sum_{j\neq i}w_{ij}(o_i-z_i)^2\\ &&+\sum_{i}w_{ii}((o_i-s_i)^2-(o_i-z_i)^2) \\
&\leq&(1+1/\epsilon)\sum_{i}\left(\sum_{j\neq i}w_{ij}(o_i-o_j)^2+w_{ii}(o_i-s_i)^2\right)\\
&&+(1+\epsilon)\sum_{i}\sum_{j\neq i}w_{ji}(o_i-z_i)^2-\sum_{i}\sum_{j\neq i}w_{ij}(o_i-z_i)^2-\sum_{i}w_{ii}(o_i-z_i)^2.
\end{eqnarray*}
The inequality follows from rearranging terms and the inequality $w_{ii}(o_i-s_i)^2 \le (1+1/\epsilon)w_{ii}(o_i-s_i)^2$.
By the condition $\sum_{j\neq i}w_{ji}/(w_{ii}+\sum_{j\neq i}w_{ij})\leq 1/(1+\epsilon)$ for all $i$,
we then have $\lambda=1+1/\epsilon$ and $\mu=0$.
By Theorem~\ref{thm:extension}, the upper bound of $\lambda/(1-\mu)$ is $1+1/\epsilon$.
\end{proof}
Our bound recovers the tight bound of $2$ for directed cycles in \cite{bindel} since $\epsilon=1$ in any directed cycle with $w_{ii}=1$ and $\sum_{j\neq i}w_{ij}=\sum_{j\neq i}w_{ij}=1$ for all $i$. Our bound also matches the upper bound of $d+1$ for $d$-regular Eulerian graphs in \cite{bindel} since $\epsilon=1/d$ in any $d$-regular Eulerian graph with $w_{ii}=1$ and $\sum_{j\neq i}w_{ij}=\sum_{j\neq i}w_{ij}=d$ for all $i$.

\begin{theorem} \label{thm:example}
There exists an example just slightly violating the conditions~(\ref{thm:condition}) in Theorem~\ref{thm:bound} with an unbounded price of anarchy.
\end{theorem}
\begin{proof}
We show the existence by a probabilistic argument. Consider a directed tree in which each node, except the root, has out-degree $1$ and this out-edge is directed to its parent node.
Let the in-degree of any node, except the leaves, at level~$i$ be a random variable $X_i \in \{1, 2\}$.
All $X_i$'s are independent from each other, and have the same expectation $E[X_i]=p$, for a constant $p$ slightly above $1$.
The root has self-loop weight $\sqrt{p}$, and all the nodes at level~$i$ have self-loop weight $\sqrt{p}-1$ and weight $1$ to its parent at level $i-1$.
Let the root be at level~$0$ and let the leaves be at level $L=\log_p n$. Let the root have internal opinion $1$ and all others have $0$.

At Nash equilibrium, the root expresses $z^{(0)}=1$ and has cost $c^{(0)}=0$, while it is easy to show by induction that every node at level $i$ expresses $z^{(i)}=p^{-i/2}$ (as $z^{(i)}=z^{(i-1)}/\sqrt{p}$) and has cost $c^{(i)}=(\sqrt{p}-1+(1-\sqrt{p})^2) p^{-i} = (p-\sqrt{p}) p^{-i}$.
As the expected number of nodes at level $i$ equals $E[\prod_{j=1}^i X_j]= \prod_{j=1}^i E[X_j] = p^i$ by independence of all $X_i$'s, the expected social cost at Nash equilibrium is $\sum_{i=1}^L p^i c^{(i)}= (p-\sqrt{p}) L$.
On the other hand, the social optimum cost is at most $\sqrt{p}$, and hence the expected price of anarchy is $\Omega(L) = \Omega(\log n)$, for a constant $p$.
Note that the tree may have nodes with in-degree $2$, and they violate the condition \eqref{thm:condition}. However, as $p$ approaches $1$, one can show by a Chernoff bound that with high probability, the fraction of such violating nodes is very small (close to $p-1$) and the price of anarchy remains $\Omega(\log n)$. This implies the existence of a weighted directed tree such that the fraction of violating nodes is very small while the price of anarchy is unbounded when $n$ is large.
\end{proof}


\section{Discussions and Future Work}
Given the results that we just present, an immediate extension may be to give more refined bounds for classes of directed graphs or various subclasses.
Another natural direction is to improve the price of anarchy with so-called ``stackelberg strategies":
one way to control some nodes is to make them express something that leads to approximate the optimal social cost or minimize the price of anarchy given the uncontrolled ones still expressing selfishly.
If one can control an arbitrary number of nodes, a straightforward way is to control the nodes that violate conditions~(\ref{thm:condition}) simply making them express their opinions at social optimum and analyze the induced equilibrium reusing Theorem~\ref{thm:bound} to at least get a bounded price of anarchy.
Theorem~\ref{thm:example} implies when we can control only $k$ nodes, $k$ has to be at least the number of nodes violating their corresponding conditions to get a bounded price of anarchy. So, if we focus on lower bounds on the size of the controlled nodes for a bounded price of anarchy,
consider $\sqrt{n}$ disconnected stars of size $\sqrt{n}$ each, where for each star all the leaf nodes (except the center) are only in influenced by the center and the center is not influenced by anyone.
In this case, at least $\sqrt{n}$ centers need to be controlled or the price of anarchy bound would be unbounded.
There may exist stricter lower bounds than this.

Note that although the results of \cite{gionis,ahmadinejad} identified some $k$-subset of nodes to control,
just as we aim to do, we have a different social cost that is the sum of all nodes' individual costs from theirs,
such as the sum of all nodes' opinions in \cite{gionis}, which happens to be ``submodular" and enables applicability of greedy algorithms there.
Other ways to control nodes may include changing their internal opinions \cite{ahmadinejad}.
More broadly, designing the networks such as deleting edges (or adding edges) to perform socially better has also been briefly discussed \cite{bindel}.

\section*{Acknowledgements}
Po-An Chen is supported in part by MOST 104-2221-E-009-045. 

\bibliographystyle{plain}

\end{document}